\documentclass[11pt]{llncs}
\newif\ifFull
\Fullfalse

\usepackage{graphicx}
\usepackage{times}
\usepackage{url}

\def\R{{\bf R}}

\makeatletter
\def\hb@xt@{\hbox to }
\makeatother

\let\oldendproof\endproof
\def\endproof{\qed\oldendproof}

\begin{document}

\ifFull
\else
\title{\LARGE Succinct Greedy Graph Drawing in the Hyperbolic Plane}
\fi

\author{David Eppstein$^1$ \and Michael T. Goodrich$^2$}

\institute{Computer Science Department, University of California, Irvine, USA.}

\maketitle   

\begin{abstract}
We describe an efficient method for 
drawing any $n$-vertex simple graph $G$ in the hyperbolic plane.
Our algorithm produces 
\emph{greedy} drawings, which support 
greedy \emph{geometric routing},
so that a message $M$ between any pair of vertices 
may be routed geometrically, simply by having each vertex that receives 
$M$ pass it along to any neighbor that is closer in the hyperbolic 
metric to the message's eventual destination.
More importantly, for networking applications,
our algorithm produces \emph{succinct} drawings,
in that each of the vertex positions
in one of our embeddings can be represented 
using $O(\log n)$ bits 
and the calculation of which neighbor to send a message to may be performed efficiently using these representations.
These properties are useful, for example, for routing in sensor networks,
where storage and bandwidth are limited.
\end{abstract}

\footnotetext[1]{\url{http://www.ics.uci.edu/~eppstein/}}
\footnotetext[2]{\url{http://www.ics.uci.edu/~goodrich/}}

\section{Introduction}
One of the richest modern applications of algorithmic graph theory is
in networking, and one of the most important algorithmic problems in networking
is \emph{routing}.
In this problem, 
we are given an $n$-vertex graph $G$ representing a communication network, 
where each vertex in $G$ is a computational agent, such as a sensor, 
smart phone, base station, PC, or workstation,
and the edges in $G$ represent communication channels.
The routing problem is to set up an efficient means to support message passing
between the vertices in $G$.

The traditional way to do routing 
is via protocols,
such as in the link-state/OSPF
or distance-vector/RIP protocols (e.g., see~\cite{c-iti-06,t-cn-03}),
that set up routing tables for each vertex $v$ in $G$.
Each such routing table has size $n$ (represented using
$\Theta(n\log n)$ bits) for each vertex $v$ in $G$, which allows $v$ to
determine to which of its neighbors it should send
a message destined for another node $w$ in $G$.
Such a solution allows for a simple message-forwarding policy, but it
is space inefficient and it requires considerable setup overhead.

There is a recent alternative approach to solving the network routing
problem, however, which can be viewed as 
new and exciting application of graph drawing.
In this alternative approach, called 
\emph{geometric routing}~\cite{bmsu-rgdah-01,kk-gpsr-00,%
fs-odgfc-06,kwz-aogma-02,kwzz-gacr-03,kwz-wcoac-03}
or \emph{geographic routing}~\cite{k-gruhs-07},
the graph $G$ is drawn in a geometric metric space $\cal S$ in the
standard way, so that vertices are drawn as points in $\cal S$ and
each edge is drawn as the loci of points along the shortest path between
its two endpoints.
For example, if $\cal S$ is the Euclidean plane, $\R^2$, then
edges would be drawn as straight line segments in this approach.
Routing is then performed by having any vertex $v$
holding a message destined for a node $w$ use a simple policy
involving only the coordinates of $v$ and $w$ and the coordinates and
topology of $v$'s neighbors to determine the neighbor of $v$ to which
$v$ should forward the message.
It is important to note that even in applications where the vertices
of $G$ come with pre-defined geometric coordinates (e.g., GPS
coordinates of smart sensors), the drawing of $G$ need not take these
coordinates into consideration, and, in fact, many known 
geometric routing schemes ignore pre-existing coordinates 
and create a new embedding using only the graph structure.
Thus, this approach to solving the routing
problem is a direct application of graph drawing.

There is an important difference between the geometric routing
problem and traditional graph drawing, however: unlike traditional
graph drawing, the main criterion for
judging an embedding done for geometric routing purposes 
is not its aesthetic qualities.
Instead, in this application,
we judge embeddings by how easily they support
simple and efficient routing protocols.

\begin{figure}[t]
\centering\includegraphics[width=1.5in]{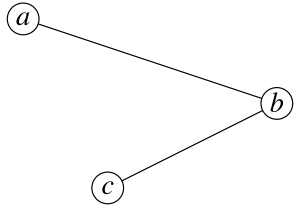}
\caption{An embedding of a graph with three vertices and two edges into the Euclidean plane that is not greedy. To route a message from $c$ to $a$ along graph edges, it must pass through the vertex $b$, which is farther from the eventual destination.}
\label{fig:nongreedy}
\end{figure} 

Perhaps the single most
simple routing policy imaginable is the \emph{greedy} one:
\begin{itemize}
\item
If a vertex $v$ receives a message $M$ with destination $w$, $v$
should forward $M$ 
to any neighbor of $v$ in $G$ that is closer than $v$ to $w$.
\end{itemize}
Thus,
we are interested in this paper in \emph{greedy} drawings of 
arbitrary graphs, that is, drawings for which greedy
routing is always successful.

Unfortunately, greedy routing doesn't always work.
For example,
it is not uncommon for geometric graph embeddings to have ``lakes''
and ``voids'' that make greedy routing impossible in some
cases~\cite{pr-ocrgr-05}; see Figure~\ref{fig:nongreedy}. 
Indeed, in any fixed-dimensional Euclidean space, a star with 
sufficiently many leaves cannot be embedded so that all paths are greedy: 
some two leaves would form an angle greater than $\pi/3$ at the star center, 
and as in the figure a route from the leaf closer to the center 
to the other leaf would not be greedy.
Thus, in order to find greedy drawing schemes for
arbitrary connected graphs, we must consider drawings in non-Euclidean spaces.

Following the formalism of Papadimitriou and Ratajczak~\cite{pr-ocrgr-05},
which was developed for Euclidean spaces,
we say that a \emph{distance decreasing path}
from $v$ to $w$ in a
geometric embedding of $G$ is a path $(v_1,v_2,\ldots,v_k)$ such that
$v=v_1$, $w=v_k$, and
\[
d(v_i,w) > d(v_{i+1},w),
\]
for $i=1,2,\ldots,k$.
A \emph{greedy embedding}\footnote{%
    Note that this formalism is equivalent to 
    the informal notion that defines a greedy embedding as one in
    which greedy routing always works.
    The formalism based on distance decreasing paths is a little
    easier to work with than this informal notion,
    however, so it is the one we use in this paper.}
of a graph $G$ in a geometric metric space
$\cal S$ is a drawing of $G$ in $\cal S$ such that a
distance decreasing path exists between every pair of vertices in $G$.

\subsection{Prior Related Work}
Early papers on geometric routing include work by 
Bose \textit{et al.}~\cite{bmsu-rgdah-01},
who extract a planar subgraph of $G$, embed it, and then route 
a message from $v$ to $w$ by marching around
the faces intersected by the line segment $vw$ using a 
subdivision traversal 
algorithm of Kranakis {\it et al.}~\cite{kranakiscompass}.
Karp and Kung~\cite{kk-gpsr-00} introduce a hybrid scheme, which
combines a greedy routing strategy with face routing.
Similar hybrid schemes were subsequently studied by several other
researchers~\cite{fs-odgfc-06,kwz-aogma-02,kwzz-gacr-03,kwz-wcoac-03}.
An alternative hybrid augmented greedy
scheme is introduced by Carlsson and Eager~\cite{ce-negrw-07}.

Rao {\it et al.}~\cite{rrpss-grli-03} introduce the idea of drawing
a graph using virtual coordinates and doing a pure greedy
routing strategy with that drawing, although they make no theoretical
guarantees.
Papadimitriou and Ratajczak~\cite{pr-ocrgr-05} 
continue this line of work on greedy drawings, 
studying greedy schemes that are guaranteed to work,
and they conjecture that Euclidean greedy drawings exist
for any graph containing a 3-connected planar spanning subgraph.
They present a greedy drawing algorithm for embedding 
3-connected planar graphs in $\R^3$ based on 
a specialization of Steinitz's Theorem for circle packings, albeit
with a non-standard metric.
Dhandapani~\cite{d-gdt-08} provides an existence proof that two-dimensional 
Euclidean 
greedy drawings of triangulations are always possible, but he does not 
provide a polynomial-time algorithm to find them.
Chen {\it et al.}~\cite{cgw-dcvc-07} study methods for producing
two-dimensional Euclidean greedy drawings for graphs containing power
diagrams, and Lillis and Pemmaraju~\cite{lp-oelia-08} provide similar methods
for graphs containing Delaunay triangulations. 
It is not clear whether either of these methods runs in polynomial time,
however. 
Thus, as of this writing, the problem of finding a polynomial-time algorithm
for producing two-dimensional greedy drawings of 3-connected planar graphs
remains open.

The corresponding two-dimensional 
problem for non-Euclidean geometries has a solution,
however, in that Kleinberg~\cite{k-gruhs-07}
provides a polynomial-time algorithm for embedding any graph in the
hyperbolic plane so as to allow for greedy routing using the standard metric
for hyperbolic space.

\subsection{The Importance of Succinctness}
Unfortunately, all of the algorithms mentioned above for producing greedy
embeddings, including the hyperbolic-space solution of
Kleinberg~\cite{k-gruhs-07}, contain a hidden drawback that makes them
ill-suited for the motivating application of geometric routing.
Namely, each of the greedy embeddings mentioned above use vertex
coordinates with representations 
requiring $\Omega(n\log n)$ bits in the worst case.
Thus, these greedy approaches to geometric routing 
have the same space usage 
as traditional routing table approaches. 
Worse, the above greedy embedding schemes have 
inferior bandwidth requirements, since they use message headers of length
$\Omega(n\log n)$ bits in the worst case, whereas traditional routing
table approaches use message headers of size $\Theta(\log n)$ bits. 
Since the \textit{raison d'{\^e}tre} for greedy embeddings is to improve and
simplify traditional routing schemes,
if embeddings are to be useful for geometric routing purposes, 
they should be \emph{succinct},
that is, they should use vertices with representations
having a number of bits that is polylogarithmic in $n$.

We are, in fact, not the first to make this observation.
Muhammad~\cite{m-adgra-07} specifically addresses succinctness,
observing that a method based on extracting a planar
subgraph of the routing network $G$ and performing a hybrid 
greedy/face-routing algorithm in this embedding 
can be implemented using only $O(\log n)$ bits for each vertex
coordinate, since planar graphs can be drawn in $O(n)\times O(n)$
grids~\cite{fpp-hdpgg-90,s-epgg-90}.

For non-Euclidean spaces, Maymounkov~\cite{m-getel-06}
provides a greedy drawing method
for three-dim\-en\-sion\-al 
hyperbolic space using vertices that can be represented with $O(\log^2 n)$ 
bits.
His work leaves open the existence of succinct greedy embeddings for
two-dimensional non-Euclidean spaces, however, as well as whether there are
succinct non-Euclidean greedy
embeddings that use only $O(\log n)$ bits per vertex.

\subsection{Our Results}
In this paper, we settle both questions of whether there are
succinct greedy embeddings in two-dimensional non-Euclidean spaces
and whether the vertices in such embeddings can be represented using an
asymptotically optimal number of bits.
In particular, we show that any $n$-vertex connected graph
can be drawn in the hyperbolic plane with coordinates that can be
represented using $O(\log n)$ bits so as to support greedy 
geometric routing between any
pair of vertices, using a standard distance metric for 
hyperbolic space.
Our scheme is constructive, runs in polynomial time, and allows the distance between any two vertices to be calculated efficiently from our representation of their coordinates.
In addition, our greedy drawing scheme is based on
the combination of a number of graph drawing and data structuring 
techniques.

\section{Autocratic Weight-Balanced Trees}
One of the new data structuring techniques we use in our
greedy drawing scheme is
a data structure that we call
\emph{autocratic weight-balanced binary trees}.
These are first and foremost weight-balanced binary trees,
which store weighted items at their leaves so that the
depth of each item of weight $w_i$ is $O(\log W/w_i)$, where $W$ is the sum
of all weights.
Just as important, however, is that they are autocratic, by which we mean
that the distance from any leaf $v$ to any other leaf $w$ is strictly greater
than the distance from the root to $w$, where tree distance is measured 
by simple path length.
Of course, this autocratic property implies that such binary trees are not
proper, in that we allow for some internal nodes 
in such trees to have only one child.
The challenge, of course, is to have a structure that is both autocratic and
weight-balanced.

\begin{figure}[t]
\centering\includegraphics[width=4in]{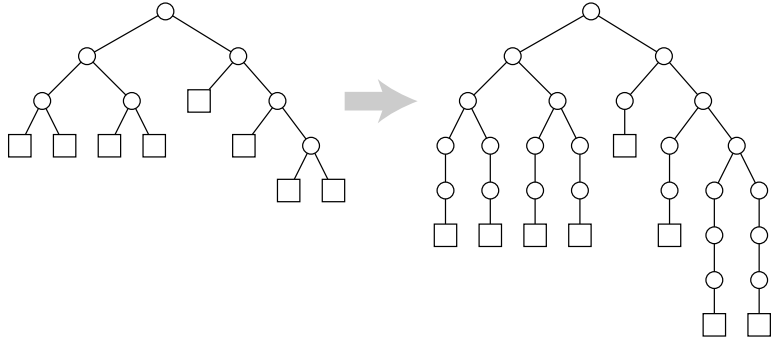}
\caption{Converting a weight-balanced binary tree into an autocratic weight-balanced binary tree.}
\label{fig:autocracy}
\end{figure}

It turns out that there is a fairly simple method for turning 
any weight-balanced binary tree into an autocratic weight-balanced
tree.
So suppose we are given
an ordered collection of $k$ items with weights $\{w_1,w_2,\ldots, w_k\}$,
such that each $w_i\ge 1$.
If we store these items at the leaves of a binary tree $T$, we say
that $T$ is \emph{weight-balanced} if the depth of each item $i$ is
$O(\log W/w_i)$,
where $W=\sum_i w_i$.
There are several existing schemes for producing 
a weight-balanced
binary tree 
so that an inorder listing of the items stored at its leaves
preserves the given order (e.g., see~\cite{GilMoo-BSTJ-59,Knu-AI-71}).

Suppose, then, that $T$ is such an ordered weight-balanced tree, and let $r$
denote the root of $T$.
To convert $T$ into an autocratic weight-balanced tree, $T'$, we replace
the edge connecting each leaf $v$ to its parent 
with a path of length
\[
1+d_T(r,{\rm parent}(v)) ,
\]
where $d_T(v,w)$ denotes the length of the path from $v$ to $w$ 
in the tree $T$.
That is, we insert a number of ``dummy'' nodes between each leaf and
its parent that is equal to the depth of its parent.
(See Figure~\ref{fig:autocracy}.)

This transformation increases the depth of each leaf in $T$ by less than a
factor of two and it keeps the depth of all other nodes in $T$
unchanged. 
Thus, if the depth of a leaf storing item $i$ in $T$
was previously at most $c\log W/w_i$, for some constant $c$, 
then the depth of the corresponding leaf in $T'$ is less than 
$2c\log W/w_i$, which is still $O(\log W/w_i)$.
Given that $T$ was weight-balanced, this implies that
$T'$ is a weight-balanced tree.
More importantly, we have the following lemma.

\begin{lemma}
The above transformation of a weight-balanced
tree $T$ produces an autocratic weight-balanced tree $T'$.
\end{lemma}
\begin{proof}
We have already observed that the tree $T'$ is weight-balanced.
So we have yet to show that $T'$ is autocratic.
First, observe that, by a simple induction argument,
if $u$ is an ancestor in $T$ of a leaf $v$, then in $T'$ we have the
following:
\[
d_{T'}(u,v) = d_T(r,v) + d_T(u,v) - 1.
\]
In particular, we have the following:
\[
d_{T'}(r,v) = 2 d_T(r,v) - 1 .
\]
Let $v$ and $w$ be two leaves in $T'$.
Furthermore, let $u$ be the least common ancestor of $v$ and $w$ in
$T'$.
Then
\begin{eqnarray*}
d_{T'}(v,w) &=& d_{T'}(u,v) + d_{T'}(u,w) \\
            &=& d_T(r,v) + d_T(u,v) - 1 + d_T(r,w) + d_T(u,w) - 1 \\
            &=& d_T(r,u) + d_T(u,v) + d_T(u,v) - 1 + d_T(r,w) + d_T(u,w) - 1 \\
            &=& (d_T(r,u) + d_T(u,w)) + d_T(r,w) + 2 d_T(u,v) - 2 \\
            &=& 2 d_T(r,w) + 2 d_T(u,v) - 2 \\
            &\ge& 2 d_T(r,w) \\
            &>& d_{T'}(r,w).
\end{eqnarray*}
Thus, $T'$ is an autocratic weight-balanced tree.
\end{proof}

Therefore, we have a way of constructing for any ordered set of weighted
items an autocratic weight-balanced tree for that set.
We will use such data structures as auxiliary components in the 
structures we discuss next.

\section{Heavy Path Decompositions}

\begin{figure}[t]
\centering\includegraphics[width=5in]{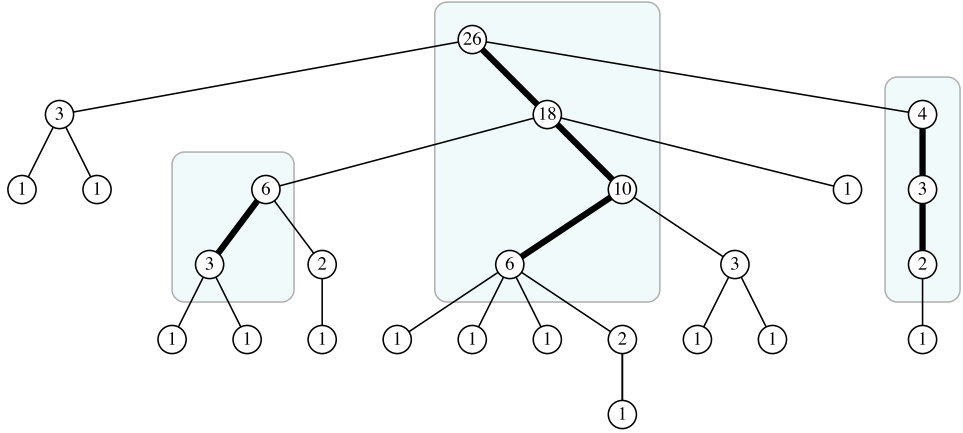}
\caption{The heavy path decomposition of a tree. Three heavy paths are shown; the remaining 17 nodes form degenerate length-0 heavy paths.The numbers shown at each node $v$ are the sizes $n(v)$ of the subtree rooted at that node.}
\label{fig:heavypath}
\end{figure}

Let $T$ be a rooted ordered tree of arbitrary degree and depth having $n$
nodes.
Sleator and Tarjan~\cite{SleTar-JCSS-83} describe a scheme, which
we call the \emph{heavy path decomposition}, 
for decomposing $T$ into a
hierarchical collection of paths (see also~\cite{SchVis-SJC-88} for an alternative path decomposition scheme with similar properties).

Their scheme works as follows.
For each node $v$ in $T$, let $n(v)$ denote the number of descendents in the
subtree rooted at $v$, including $v$ itself.
For each child-to-parent edge, $e=(v,w)$ in $T$, label
$e$ as a \emph{heavy} edge if $n(v) > n(w)/2$.
Otherwise, label $e$ as a \emph{light} edge.
Connected components of heavy edges form paths, called
\emph{heavy paths}, which may in turn have many incident light edges.
As a degenerate case, we also consider the zero-length path consisting of a
single node in $T$ incident only to light edges as a heavy path.

Note that the size of a subtree at least doubles every time we traverse a light
edge from a child to a parent. 
(See Figure~\ref{fig:heavypath}.)
Thus, if we compress every heavy path in $T$ to a
single ``super'' node, preserving the relative order of the nodes, then we
define a tree, $Z$, of depth $O(\log n)$.

Of course, the nodes in $Z$ can have arbitrary degree. Nevertheless, 
for data structuring purposes, following Alstrup et al.~\cite{AlsLauSom-WADS-97}, we may replace each vertex $v$ in $Z$ having
$d$ children $v_1,v_2,\ldots,v_d$ with a weight-balanced binary tree
that uses the $n(v_i)$ values as weights.
The useful property of this substitution is that
any leaf-to-root path $P$ in the resulting binary tree, $Z''$,
will have length $O(\log n)$, since the lengths of the subpaths 
of $P$ in the weight-balanced binary trees traversed in $P$ form 
a telescoping sum that adds up to $O(\log n)$.
That is, it can be written as a value proportional to
something of the form

\vspace*{-20pt}
\begin{eqnarray*}
& & \log n_0 + \log n_1/n_0 + \log n_2/n_1 + \cdots + \log n/n_k \\
& =& \log n_0 + \log n_1-\log n_0 +\log n_2-\log n_1 +\cdots +\log n-\log n_k \\
& =& \log n_0 + \log n,
\end{eqnarray*}

\vspace*{-8pt}\noindent
where $n_0$ is a constant.

In our case, we use autocratic weight-balanced binary trees for the
substitutions of high-degree super nodes in $Z$, so as to define the binary
tree of depth $O(\log n)$.
This construction will prove essential for our greedy embedding scheme.
Before we present this geometric embedding, however, we first present 
a combinatorial greedy embedding in a completely contrived metric space,
which we will subsequently show how to turn into a greedy embedding
in the hyperbolic plane using the standard hyperbolic metric.

\section{Greedy Embeddings in the Dyadic Tree Metric Space}

\begin{figure}[t]
\centering
\includegraphics[width=2.5in]{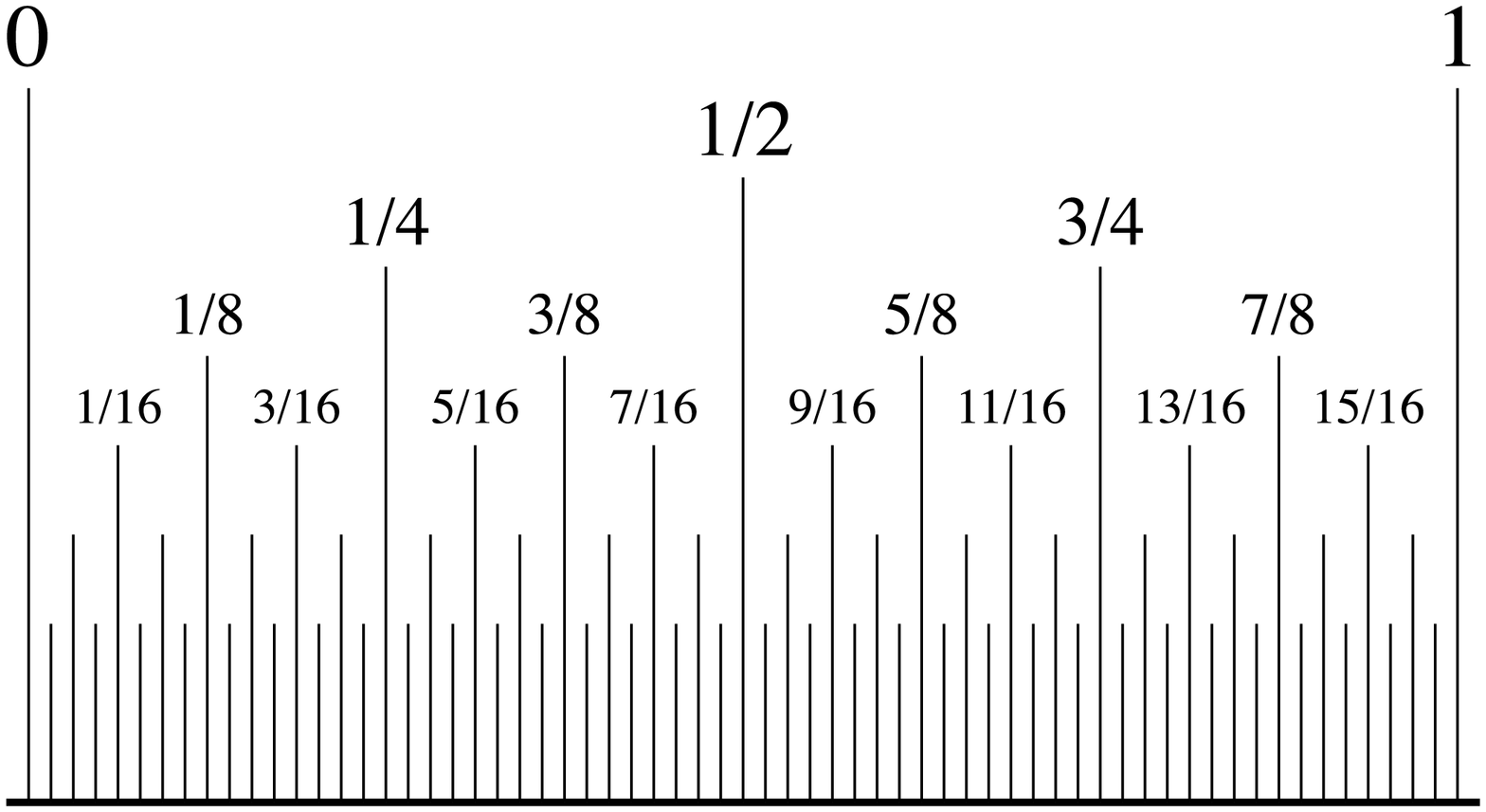}
\qquad\qquad
\includegraphics[width=2.5in]{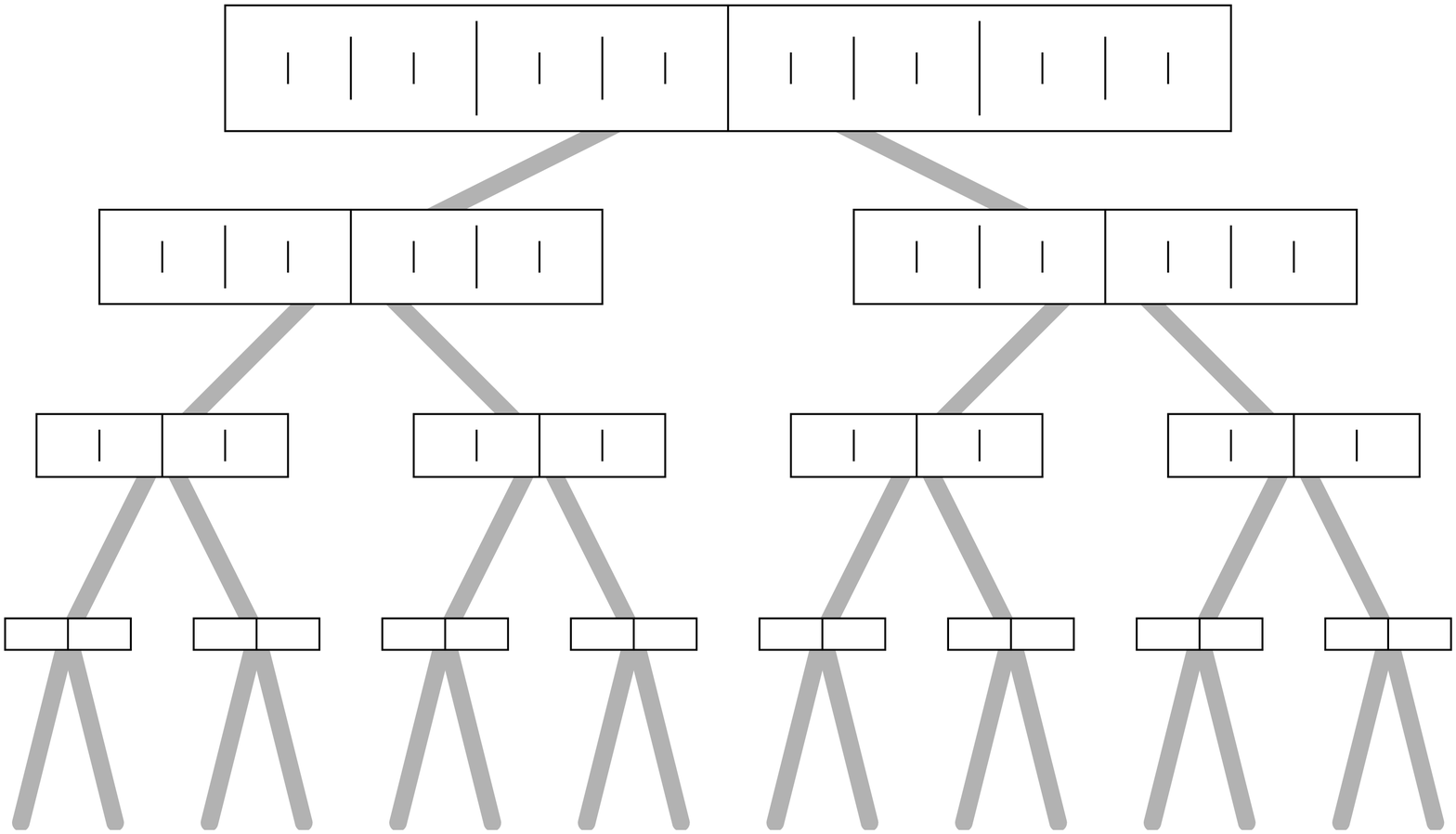}
\caption{The dyadic rational numbers (left) and a schematic view of the dyadic tree metric space (right).}
\label{fig:dyadic}
\end{figure}

Suppose we are given a graph $G$ having $n$ vertices and $m$ edges
for which we wish to 
construct a succinct greedy embedding.
We show in this section how to produce a combinatorial greedy embedding in a
contrived space we call the \emph{dyadic tree metric space}.

We may consider an infinite binary tree to be an abstract metric space, in which the distance between any two tree nodes is just the number of edges on the shortest path between them. But there is another natural metric that can be formed on the same tree by embedding it into the \emph{dyadic rational numbers} (Figure~\ref{fig:dyadic}, left), rational numbers with denominators that are powers of two. Let $f$ be the map from the infinite binary tree to the open interval $(0,1)$ that maps the root of the tree to $1/2$, and that maps the children of a node $x$ at level $i$ of the tree to $f(x)\pm 2^{-i-2}$; thus, the children of the root map to the dyadic rational numbers $1/4$ and $3/4$, the grandchildren of the root map to $1/8$, $3/8$, $5/8$, $7/8$, etc. We define the \emph{dyadic metric} on the infinite binary tree as the metric in which the distance between two tree nodes $x$ and $y$ is $|f(x)-f(y)|$. Note that all distances in the dyadic metric are less than one.

We will show that any graph may be greedily embedded into a hybrid ad-hoc metric space that combines features from both of these two tree metrics; we call it the \emph{dyadic tree metric space}.  A point in this space is represented by a pair $(x,y)$, where $x$ and $y$ are nodes in the infinite binary tree and where $x$ must be an ancestor of $y$ (possibly equal to $y$ itself). We define the distance between two points $(x,y)$ and $(x',y')$ in the dyadic tree metric space to be the sum of the tree distance between $x$ and $x'$ and of the dyadic distance $|f(y)-f(y')|$. The dyadic tree metric space can be represented as an infinite binary tree representing the $x$ coordinates of each of its points, in which each tree node contains an interval of dyadic rational numbers; this interval of numbers is split into two halves at the two children of each node. This representation is depicted in Figure~\ref{fig:dyadic}, right.

\begin{figure}[t]
\centering\includegraphics[width=6in]{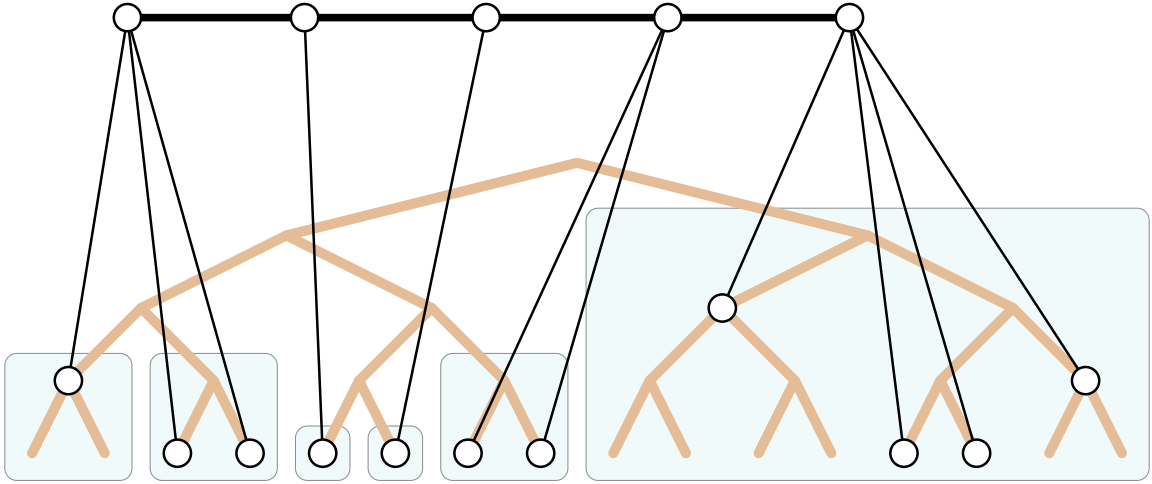}
\caption{Our two-level weight-balanced strategy for placing the children of the nodes on a heavy path. The groups of children for each heavy path node are assigned to subtrees in a weight-balanced way (gray shaded areas), and then within each subtree the individual children are placed using a second level of weight balancing. The third step of child placement, in which we make the subtree between the root (representing the heavy path) and its children autocratic, is not shown.}
\label{fig:double-balance}
\end{figure}

Our embedding begins with us finding a spanning tree $T$
of $G$, choosing a root arbitrarily, and producing a heavy path decomposition of $T$. For technical reasons we require that each node in a nontrivial heavy path of the decomposition have at least one child that is not in the path; we add dummy nodes to $T$ if necessary, after forming the path decomposition, to ensure that this is true. There will be fewer than $n$ dummy nodes added, so they will not significantly increase the number of bits needed to represent each vertex in our greedy embedding.

We orient the light edges for each heavy path $P$ so that they are 
all on the same side of $P$ and we orient the light edges incident upon the
same vertex. 
We then compress each
heavy path into a super node, using the orientation of edges around the
vertices of each heavy path to determine the ordering of children for each
node in the resulting tree, $Z$. If a super node in $Z$ is the right child of its parent, we make the left-to-right ordering of children be the same as the ordering from parent to child in the heavy path; if, on the other hand, it is the left child of its parent, we make the left-to-right ordering of children be the same as the ordering from child to parent in the heavy path.

Next, we form groups of the nodes in $Z$ that have the same parent in $T$.
We form a weight-balanced binary tree for each these groups.
Furthermore,
within each group, we form a weight-balanced binary 
tree of the nodes in the group.
Concatenating these two levels of weight-balanced trees forms a single
weight-balanced tree connecting the node in $Z$ to each of its
children; we apply the transformation described earlier to make this
tree autocratic. The first three steps, in which we form a
weight-balanced tree of the groups and a weight-balanced tree within
each group, and then concatenate these two levels of weight-balanced
trees to form a single binary tree for all children of the node in $Z$,
are depicted in Figure~\ref{fig:double-balance}.

This construction of an autocratic weight-balanced tree for each node
in $Z$ can be used to embed $Z$ as a whole into the infinite binary
tree. The root of $Z$ may be placed at the root of the infinite binary
tree, and the children of each node $v$ in $Z$ are placed under that
node in the positions of the infinite binary tree corresponding to
their positions in the autocratic weight-balanced tree constructed for
$v$.  We observe that, in this way, all nodes of $Z$ are placed at most
$O(\log n)$ levels deep in the infinite binary tree; for, due to the
weight balancing, the distance in the infinite binary tree between any
node $w$ and its parent $v$ is proportional to the difference in the
logarithms of the weights of the subtrees rooted at $v$ and $w$, and
along any path of $Z$ these differences add in a telescoping series to
$O(\log n)$.

We have embedded $Z$ into the infinite binary tree; we are now ready to embed $T$ itself into the dyadic tree metric. To do so, we must determine a pair $(x,y)$ of coordinates for any node $v$ of $T$; both $x$ and $y$ must be nodes of the infinite binary tree, and $x$ must be an ancestor of $y$. The $x$ coordinate of $v$ is simply the node of the infinite binary tree at which the heavy path of $v$ is placed. The $y$ coordinate of $v$ is the least common ancestor in the infinite binary tree of the placements of all the children of $v$. This calculation is the reason we required $v$ to have at least one child; for leaf nodes of $T$, we instead set $y=x$. Due to our two-level weight balancing strategy, two nodes of $T$ that belong to the same heavy path (and that therefore share the same $x$ coordinate) will have different $y$ coordinates, for their children will be placed within disjoint subtrees of the infinite binary tree.

\begin{lemma}
The embedding of $T$ into the dyadic tree metric space described above is greedy.
\end{lemma}

\begin{proof}
Any directed path in $T$ consists of edges that, when translated into the dyadic tree metric space, have three types: edges from a node to the parent heavy path in $Z$, edges within a heavy path, and edges from a node to a child heavy path in $Z$. We must show that edges of each type lead to a node that is closer to the terminus of the path.

For the edges that go from a node to the parent heavy path or to a child heavy path, this is straightforward: the contribution of the $x$-coordinates to the distance to the terminus decreases by one at each step, due to the autocratic property of our weight-balanced trees, more than offsetting any possible increase in the contribution of the $y$-coordinates.

For the edges that remain within a heavy path, the $x$ coordinates remain unchanged and do not lead to any increase or decrease of the distance to the terminus. The $y$ coordinates are linearly ordered by the map $f$ from infinite binary tree nodes to dyadic rationals, and our weight-balanced trees were chosen to be consistent with this linear ordering; therefore, any step along the heavy path, either towards a node of the path that is the ancestor of the terminus or towards the topmost node of the path and the edge leading to the parent node in $Z$, decreases the distance to the terminus.
\end{proof}

As in previous work~\cite{k-gruhs-07}, a greedy embedding for the spanning tree $T$ is automatically greedy for the overall graph $G$ from which it was drawn.

\section{Succinct Greedy Embedding in the Hyperbolic Plane}

We have shown that any tree $T$ (and any graph $G$ by choosing a spanning tree of $G$) may be greedily and succinctly embedded into a dyadic tree metric space. To complete our greedy embedding, it remains to show that this space may be embedded, independently of our original graph (but depending on a parameter $D$ determined by the number of vertices of the graph), into the hyperbolic plane in such a way that the greedy property of the embedding of $T$ is preserved.  That is, although the distances themselves in the hyperbolic plane may differ from those in the dyadic tree metric space, composing our embedding of $T$ into the dyadic tree metric space with our embedding of the dyadic tree metric space should yield a greedy embedding of $T$ into the hyperbolic plane.

Due to the existence of this embedding, we may reinterpret the succinct coordinates computed for the embedding of a graph into the dyadic tree metric space as also being coordinates for a subset of points in the hyperbolic plane. Not every hyperbolic plane point will be representable with such coordinates, but this is no different in principle from using pairs of integers to represent grid points in the Euclidean plane: not every Euclidean point is representable as an integer grid point. The parameter $D$ is analogous to the scale of a grid embedding.

\begin{figure}[t]
\centering\includegraphics[width=4in]{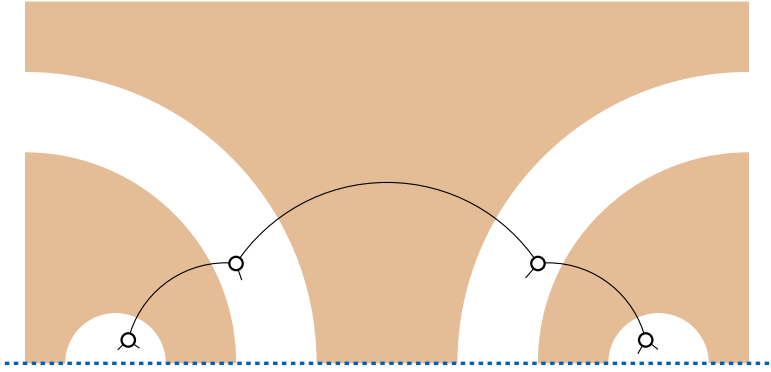}
\caption{Disjoint buffer zones of width $D$ are crossed by each edge of an embedding of the infinite binary tree into the hyperbolic plane, so that the tree distance and hyperbolic distance closely approximate each other.}
\label{fig:buffers}
\end{figure}

Our overall strategy will be to embed the infinite binary tree into the hyperbolic plane in such a way that any edge has length $D+O(1)$ and crosses a \emph{buffer zone} of width $D$, bounded by two hyperbolic lines (Figure~\ref{fig:buffers}). The buffer zones for different edges will be disjoint from each other. Thus, any two nodes of the tree that have tree distance $k$ units apart will have hyperbolic distance at least $Dk$ (because any path between the two nodes must cross $k$ buffer zones) and at most $(D+O(1))k$ (there exists a path following tree edges with that length). In our application, all tree paths will have $O(\log n)$ edges; thus, by choosing $D=\Omega(\log n)$ we may guarantee that the order relation between any two distinct tree distances remains unchanged by this hyperbolic embedding. Any point $(x,y)$ of the dyadic tree metric will be placed near the embedding of tree node $x$, and this placement will ensure the greediness of any edge whose endpoints belong to different paths of our heavy path decomposition.

\begin{figure}[t]
\centering\includegraphics[width=4in]{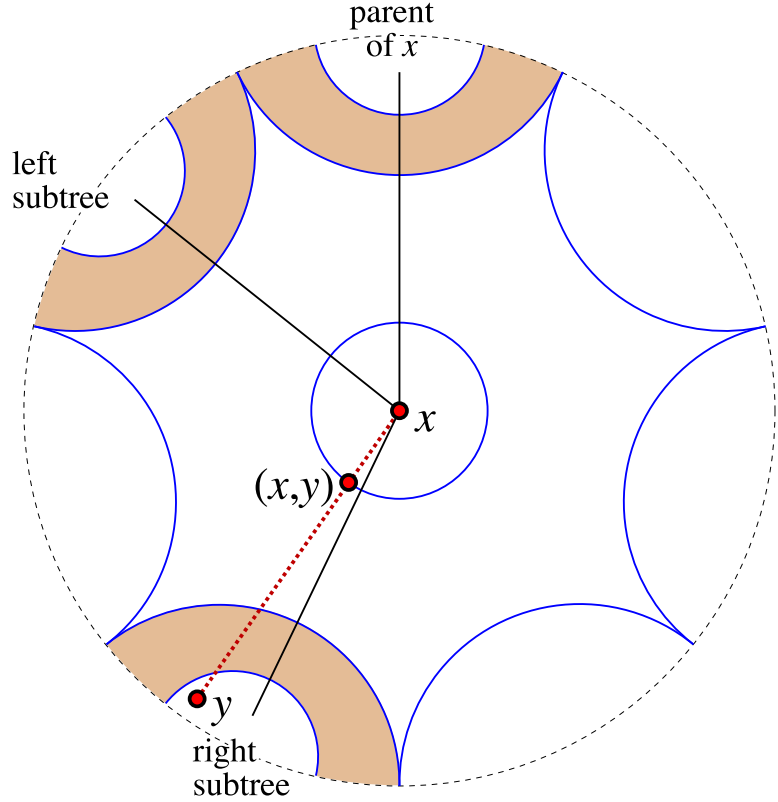}
\caption{Top-down placement of node $x$ of the infinite binary tree and point $(x,y)$ of the dyadic tree metric space into the hyperbolic plane, shown in a Poincar\'e disk model centered at $x$.}
\label{fig:heptagon}
\end{figure}

Next, we place nodes of the infinite binary tree into the hyperbolic plane, with the buffer zones described above. Although this placement is conceptual rather than algorithmic, we may view it as being performed in a top down traversal of the tree, so that when node $x$ is placed we will already know the location of its parent, the buffer zone separating $x$ from its parent, and a line connecting it to its parent and on which it must be placed. We place $x$ itself on this line in such a way that the boundary of the parental buffer zone forms one of the seven sides of an ideal regular heptagon---a figure in the hyperbolic plane formed by seven lines that are asymptotic to each other but never intersect, such that the angle subtended by each line as viewed from $x$ is equal. Figure~\ref{fig:heptagon} shows this placement, in a Poincar\'e disk model of the hyperbolic plane centered at $x$; the parental buffer zone is the topmost shaded region in the figure and the vertical line through $x$ is the one connecting it to its parent node. The large arcs depict hyperbolic lines forming the heptagon described above.

In the case where $x$ is the right child of its parent, so that the upper nodes of the heavy path represented by $x$ have children in its left subtree and the lower nodes of the heavy path have children in the right subtree, shown in the figure, we place the left subtree within the halfplane bounded by the heptagon side one step counterclockwise from the parent, and the right subtree within the halfplane bounded by the heptagon side three steps counterclockwise from the parent, as shown in the figure. In the case where $x$ is its parent's left child, we reverse the figure, placing the right subtree within the halfplane one step clockwise from the parent and the left subtree within the halfplane three steps clockwise from the parent. In either case, we draw lines connecting $x$ to its child nodes, at angles of $2\pi/7$ and $6\pi/7$ from the angle of the line connecting $x$ to its parent (the solid straight lines of the figure). We use the heptagon edges as the outer boundaries of buffer zones between $x$ and its children, and we set the inner boundaries of the buffer zones to be hyperbolic lines perpendicular to the lines connecting $x$ to its children, at distance $D$ from the outer boundaries of the buffer zones. With this information determined, we may continue to place the children of $x$ in the same way.

We are finally ready to describe the mapping of the dyadic tree metric space into the hyperbolic plane. Recall that each point of the dyadic tree metric space consists of a pair $(x,y)$ where $x$ and $y$ are nodes of the infinite binary tree, $x$ a parent of $y$. We draw small circles of equal radius centered at each point where we have placed a node of the infinite binary tree---the precise radius is unimportant as long as it is small enough that the circles are disjoint from the buffer zones. Then, given a point $(x,y)$ of the dyadic tree metric space, we draw a hyperbolic line segment from $x$ to $y$ (the dotted straight line in the figure), and place $(x,y)$ at the point where this line segment intersects the circle centered at $x$. In the case $x=y$, which happens in our construction only for leaves, we instead place $(x,x)$ at the point where the line segment from $x$ to its parent intersects the circle centered at $x$.

\begin{figure}[t]
\centering\includegraphics[width=3.5in]{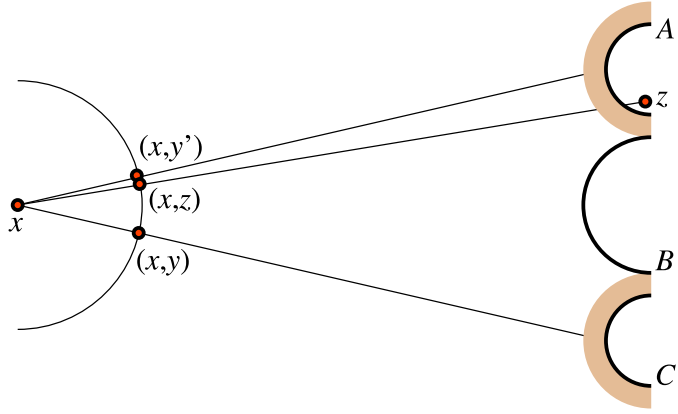}
\caption{Illustration for proof of greediness of our embedding (not to scale).}
\label{fig:heavy-is-greedy}
\end{figure}

\begin{theorem}
For sufficiently large values of $D$, the embedding of $G$ formed by composing the embedding from $G$ into the dyadic tree metric space and the embedding of the dyadic tree metric space into the hyperbolic plane is greedy.
\end{theorem}

\begin{proof}
We show that, for every edge $e$ of the chosen spanning tree, and every possible terminus $v$ of a path using $e$, that traveling along $e$ reduces the distance to the terminus. We assume that the starting endpoint of $e$ is placed at point $(x,y)$ of the dyadic tree metric, the ending endpoint is placed at point $(x',y')$, and that these points are mapped as described above to the hyperbolic plane. We distinguish several cases.

First, if $x\ne x'$, let $k=O(\log n)$ be the tree distance from $x'$ to the destination. Then, due to the autocratic property of our weight-balanced placement of heavy paths into the dyadic tree metric, $x$ is at tree distance at least $k+1$ from the destination. As discussed above, due to the buffer zones of our construction, $(x,y)$ is at hyperbolic distance at least $(k+1)D$ from the destination, while $(x',y')$ is at hyperbolic distance at most $k(D+O(1))$. By choosing $D$ sufficiently large (a constant times $\log n$), we can guarantee that the former distance is larger than the latter and that this step is greedy.

Second, if $x=x'$ and the eventual destination also has the same value of $x$, the result follows from the fact that our embedding places the nodes of any heavy path consecutively over an arc of less than half of a circle. Such an embedding is greedy for any path, no matter how the nodes are distributed within the arc.

Third, if $x=x'$ and the eventual destination is reached via the parent of $x$, the step is greedy for the same reason as in the second case: the nodes that are mapped to $x$ form a heavy path placed in order along an arc of less than half the circle, with the node of the arc closest to the parent being the apex of the heavy path.

The most complicated case is the fourth: $x=x$ and the eventual destination $z$ has $x''$ as a proper descendant of $x$. The closest point to $z$ on the circle surrounding $x$ onto which $(x,y)$ and $(x',y')$ are both mapped is the hyperbolic point represented by the coordinates $(x,z)$; the distance to $z$ from other points on the circle can be calculated as a monotonic function of the arc length between those other points and $(x,z)$. Thus, moving around the circle towards $(x,z)$ is a greedy step. Unfortunately, the point $(x,z)$ may not be a node of the heavy path; rather, the node of the heavy path from which $z$ descends may be some other nearby point $(x,y'')$. We must show that any step along the heavy path towards this point is greedy.

In most cases, it is straightforward to show that this step is greedy: a step around the circle towards $(x,z)$ is also a step towards $(x,y'')$, which as we have argued immediately above is greedy. The only possible exception occurs when $y'=y''$ and when the true closest point on the circle to $z$, that is, $(x,z)$, lies on the arc of the circle between $y$ and $y'$. In this case we must show that $(x,y')$ and $(x,z)$ are closer in arc length than $(x,y)$ and $(x,z)$, for then the greediness of the step will follow from the monotonicity of the distance to $z$ as a function of arc length.

Let $\hat y$ be the least common ancestor in the binary tree of the two disjoint subtrees containing $y$ and $y'$. Let $A$ be the inner boundary of the buffer zone adjacent to $\hat y$ that contains $y'$, let $C$ be the inner boundary of the buffer zone adjacent to $\hat y$ that contains $y$, and let $B$ be the edge of the regular ideal heptagon adjacent to $\hat y$ that separates $A$ from $C$. Figure~\ref{fig:heavy-is-greedy} illustrates this notation.
These three hyperbolic lines may not be symmetrically placed relative to $x$, due to the asymmetry of the placement of the two subtrees relative to the parent at each node $x$. However, the distances from $x$ to $A$ and to $C$ are within $O(1)$ of each other, and $B$ is closer to $x$ by a distance of $D-O(1)$. It is a basic property of hyperbolic geometry that the angle that an object subtends, as viewed from a fixed point of view $x$, is inversely proportional to an exponential function of the distance of the object from $x$. Thus, $B$ will subtend an angle, as viewed from $x$, that is larger than the angles subtended by $A$ and $C$ by a factor exponential in $D-O(1)$. In particular, for sufficiently large $D$ (larger than some fixed constant, a weaker requirement than the one above that $D=\Omega(\log n)$), both $A$ and $C$ will subtend smaller angles than the angle subtended by $B$. Then, any point behind line $A$, and in particular the point $z$, will form an arc from $(x,y')$ to $(x,z)$ that is shorter than the arc from $(x,y)$ to $(x,z)$. The greediness of the step from $(x,y)$ to $(x,y')$ follows from the monotonicity of the distance to $z$ as a function of arc length.
\end{proof}

\raggedright
\bibliographystyle{abbrv}
\bibliography{geom,hyper}

\end{document}